\documentclass[journal,twoside,web]{ieeecolor}

\usepackage{setspace}
\usepackage{url}
\usepackage[utf8]{inputenc}
\usepackage[T1]{fontenc}
\usepackage{graphicx}		
\usepackage{wrapfig}
\usepackage{sidecap} 
\usepackage[export]{adjustbox}
\usepackage[caption=false]{subfig}
\usepackage{float}

\usepackage{amsmath} 
\usepackage{amssymb}  
\usepackage{amsthm}
\usepackage{mathtools}
\usepackage[normalem]{ulem}
\usepackage{paralist}	
\usepackage[space]{grffile} 
\usepackage{color}
\let\labelindent\relax
\usepackage{enumitem}
\usepackage{bm}
\usepackage{bbm}
\usepackage{cancel}
\usepackage{hhline}
\usepackage[c2 , nocomma]{optidef}
\usepackage{oubraces}
\usepackage{algorithmic}
\usepackage{graphicx}
\usepackage{textcomp}

\usepackage{amsfonts}
\usepackage{lipsum}

\usepackage{cite}
\usepackage[hidelinks]{hyperref}

\newtheorem{theorem}{Theorem}

\newtheorem{lemma}{Lemma}
\theoremstyle{definition}
\newtheorem{definition}{Definition}
\theoremstyle{remark}
\newtheorem{remark}{Remark}
\theoremstyle{definition}
\newtheorem{assumption}{Assumption}
\theoremstyle{definition}
\newtheorem{proposition}{Proposition}
\newtheorem{example}{Example}

\newtheorem{problem}{Problem}

\newcommand{\R}{\mathbb{R}}

\definecolor{darkblue}{RGB}{0,0,102}
\definecolor{lightblue}{RGB}{77,77,148}

\definecolor{gold}{RGB}{234, 170, 0}
\definecolor{metallic_gold}{RGB}{139, 111, 78}

\newcommand{\defeq}{\triangleq}

\newcommand{\bs}[1]{\boldsymbol{ #1 }}

\newcommand{\norm}[1]{\left\Vert #1 \right\Vert}

\newcommand{\derp}[2]{\tfrac{\partial #1 }{\partial #2 }}

\usepackage{xfrac}

\usepackage{dblfloatfix}
\allowdisplaybreaks
\usepackage{pdfpages}
\usepackage{graphicx}

\usepackage{enumitem}

\newcommand{\argmin}{\operatornamewithlimits{arg~\!min}}

\newcommand{\subjectto}{\operatorname{s.t.}}

\newcommand{\bzero}{\mathbf{0}}

\newcommand{\eye}{\boldsymbol{I}}

\newcommand{\Cs}{\mathcal{C}_{\rm S}}
\newcommand{\Cb}{\mathcal{C}_{\rm B}}

\newcommand{\Cdhat}{\widehat{\mathcal{C}}_{\rm I}(t)}
\newcommand{\Tt}{T}
\newcommand{\ub}{\boldsymbol{k}_{\rm b}}
\newcommand{\up}{\boldsymbol{k}_{\rm p}}
\newcommand{\thetahat}{\hat{\boldsymbol{\theta}}}
\newcommand{\thetatilde}{\tilde{\boldsymbol{\theta}}}
\newcommand{\Thetatilde}{\widetilde{\Theta}(t)}

\newcommand{\varphibtrue}[2]{\boldsymbol{\varphi}_{\rm b}(#1,#2)}
\newcommand{\varphibhat}[2]{\widehat{\boldsymbol{\varphi}}_{\rm b}(#1,#2,t)}

\newcommand{\Phihat}{\boldsymbol{\Phi}(\tau,\boldsymbol{x},t)}
\newcommand{\PhihatT}{\boldsymbol{\Phi}(\Tt,\boldsymbol{x},t)}
\newcommand{\Gammahat}{\boldsymbol{\Gamma}(\tau,\boldsymbol{x},t)}
\newcommand{\GammahatT}{\boldsymbol{\Gamma}(\Tt,\boldsymbol{x},t)}

\newcommand{\bx}{\mathbf{x}}

\usepackage[utf8]{inputenc}
\usepackage[T1]{fontenc}
\usepackage{amsmath,amssymb,mathtools,amsfonts,bm,bbm}
\usepackage{graphicx}
\usepackage{cite}
\usepackage{color}
\usepackage{url}
\usepackage{xfrac}
\let\labelindent\relax
\usepackage{enumitem}
\usepackage{setspace}
\usepackage{dblfloatfix}
\allowdisplaybreaks
\usepackage{arydshln}
\makeatletter
\def \thm@space@setup {
    \thm@preskip = 0pt
    \thm@postskip = 0pt
}
\g@addto@macro \th@remark {
    \thm@preskip = 0pt
    \thm@postskip = 0pt
}
\makeatother

\theoremstyle{definition}

\usepackage{lcsys}

\begin{document}

\def\BibTeX{{\rm B\kern-.05em{\sc i\kern-.025em b}\kern-.08em
    T\kern-.1667em\lower.7ex\hbox{E}\kern-.125emX}}
\markboth{\journalname, VOL. XX, NO. XX, XXXX 2017}
{Author \MakeLowercase{\textit{et al.}}: Preparation of Papers for IEEE Control Systems Letters (August 2022)}

\title{ Robust Adaptive Backup Control Barrier Functions}

\author{ Ersin Da\c{s}$^{1}$, David E. J. van Wijk$^{2}$, Tamas G. Molnar$^{3}$, Aaron D. Ames$^{2}$, and Joel W. Burdick$^{2}$ 
\thanks{$^{1}$ E. Da\c{s} is with the Department of Mechanical, Materials, and Aerospace Engineering, Illinois Institute of Technology, Chicago, IL 60616, USA,  {\texttt{edas2@illinoistech.edu}}.
}
\thanks{$^{2}$ D. E. J. van Wijk, A. D. Ames, and J. W. Burdick are with the Mechanical and Civil Engineering, California Institute of Technology, Pasadena, CA 91125, USA, \texttt{\{vanwijk, ames, jburdick\}@caltech.edu}.}
\thanks{$^{3}$ T. G. Molnar is with Mechanical Engineering, Cleveland State University, Cleveland, OH 44115, USA, \texttt{t.g.molnar@csuohio.edu}.}
}

\maketitle
\thispagestyle{plain}   

\begin{abstract}
We propose a notion of robust adaptive backup control barrier functions for nonlinear control affine systems with parametric uncertainty in both the drift dynamics and actuation matrix. Backup control barrier functions guarantee safety by predicting the system's trajectory under a pre-certified safe controller. However, these predictions rely on the model and can be inaccurate when the system contains unknown parameters. To address this issue, we estimate the unknown parameters using element-wise certified adaptive estimators that provide a parameter adaptation law and component-wise estimation error bounds. We compute the backup flow using the estimated model and tighten the safety conditions using these certified bounds. The resulting safety conditions account for the sensitivity of the predicted flow to parameter estimation errors. Moreover, to handle uncertainty in the actuation matrix, we use a duality-based reformulation that enables the use of a computationally efficient quadratic-program-based safety filter. We prove that controllers satisfying the proposed robust adaptive backup control barrier function constraints guarantee safety under parametric uncertainty and input constraints. 
\end{abstract}

\begin{IEEEkeywords}
Constrained control, robust nonlinear control, adaptive control, control barrier functions
\end{IEEEkeywords}

\section{Introduction}
\label{sec:intro}
Safety-critical control is essential for autonomous systems that operate in complex real-world environments. Control barrier functions (CBFs) provide a constructive framework for synthesizing controllers that render a prescribed safe set forward invariant, ensuring safety w.r.t.~state constraints \cite{ames2017control}. This can be realized by safety filters, often formulated as quadratic programs (QPs), that modify a nominal controller when safety necessitates intervention. However, CBF design requires a safe set in which feasible safe control inputs are available \cite{isaly2024feasible}. Constructing such a set is difficult for input-constrained nonlinear systems, especially when the system must be steered close to the safe set boundary \cite{zeng2021safety, freire2026using}.

Model predictive control (MPC) provides another way to handle state and input constraints. It predicts the system evolution over a finite horizon while enforcing constraints along the predicted trajectory. This makes MPC effective for constrained nonlinear systems, but recursive feasibility and safety guarantees require additional terminal conditions, invariant sets, or robustness margins. Adaptive MPC and robust sensitivity-based MPC reduce conservatism by updating model information online \cite{sasfi2023robust, culbertson2021, minniti2021} or by propagating the effect of parametric uncertainty along the prediction horizon \cite{belvedere2025}.

Backup control barrier functions (bCBFs) use a similar finite horizon prediction idea for safety-critical control to construct an implicit safe set online \cite{gurriet_scalable_2020}. In particular, bCBFs predict the evolution of the system under a pre-certified safe controller, called a backup controller, over a finite time horizon. A system is considered safe if the predicted backup trajectory satisfies the state constraints and reaches a known backup set \cite{gurriet_scalable_2020}. This approach guarantees feasibility for input-constrained systems by providing a constructive certificate that the state remains safe. However, model uncertainty can cause the computed backup trajectory to differ from the true trajectory, invalidating safety \cite{van2024disturbance}.

This limitation becomes more severe when the uncertainty appears in both the drift dynamics and the actuation matrix. Actuation uncertainty complicates robust safety filter design because the worst-case uncertainty involves the control input that is the decision variable of the underlying optimization problem. Recent robust CBF methods address related issues by using worst-case bounds, set-membership identification, or duality-based reformulations \cite{cohen2022robust, tarun2025}. Adaptive CBF methods reduce conservatism by estimating uncertain parameters online \cite{lopez2023unmatched,cohen2024}. However, robust bCBFs require more than a pointwise estimate of uncertain parameters. Since the safety condition depends on the predicted backup flow, it must account for the sensitivity of the flow to the parameter estimation errors. Moreover, if the estimator returns only a set-valued update, it may not provide the parameter estimate required for computing the backup flow sensitivity.

In this letter, we propose a robust adaptive bCBF framework for nonlinear systems with structured parametric uncertainty in the drift dynamics and actuation matrix. We consider a parameter affine uncertainty model and use element-wise certified adaptive estimators. These estimators provide a parameter adaptation law and certified component-wise bounds on the parameter estimation error. The proposed method computes the backup flow using the estimated model and tightens the safety conditions using the certified error bounds. These tightened conditions define an inner approximation of the implicit safe set. To handle uncertainty in the actuation matrix, we use a duality-based robustification that preserves a QP-based safety filter. We demonstrate the effectiveness of the proposed approach through a planar quadrotor example. 

\section{Preliminaries}
\label{sec:pre}
We consider a nonlinear control affine system of the form
\begin{equation}
\label{eq:nom_sys}
    \dot{\bs{x}}
    =
    \bs{f}(\bs{x})
    +
    \bs{g}(\bs{x})\bs{u},
    \ \ 
    \bs{x} \in \mathcal{X} \subset \R^n,
    \ \ 
    \bs{u} \in \mathcal{U} \subset \R^m,
\end{equation}
where ${\bs{f}:\mathcal{X}\to\R^n}$ and ${\bs{g}:\mathcal{X}\to\R^{n\times m}}$ are continuously differentiable, and ${\mathcal{U}}$ is a convex polytope. 
To capture state constraints that characterize safety, we define the safe set
\begin{equation} \label{eq:safeset}
    \Cs \defeq \{\bs{x}\in\mathcal{X}: h(\bs{x}) \geq 0\},
\end{equation}
where ${h:\mathcal{X}\to\R}$ is continuously differentiable.
Safety requires that \eqref{eq:nom_sys} evolves inside the safe set, i.e., $\Cs$ must be forward invariant.
Control barrier functions (CBFs) enable the synthesis of safe controllers. The function ${h}$ is a CBF if there exists an extended class-$\mathcal{K}_{\infty}$ function $\alpha$ such that
\begin{equation} 
\label{eq:CBF}
    \sup_{\bs{u}\in\mathcal{U}}
    \big [ \nabla h(\bs{x})
    \big(
        \bs{f}(\bs{x})
        +
        \bs{g}(\bs{x})\bs{u}
    \big) \big ]
    >
    -\alpha\big(h(\bs{x})\big),
\end{equation}
holds  for all ${\bs{x} \in \Cs}$.
Then, any locally Lipschitz controller ${ \bs{u} = \bs{k}(\bs{x})}$,  ${\bs{k}: \mathcal{X} \to \mathcal{U}}$, satisfying
\begin{equation} 
\label{eq:CBF_constraint}
    \nabla h(\bs{x})
    \big(
        \bs{f}(\bs{x})
        +
        \bs{g}(\bs{x})\bs{u}
    \big)
    \geq
    -\alpha\big(h(\bs{x})\big),
    \ \
    \forall\bs{x}\in\Cs,
\end{equation}
renders ${\Cs}$ forward invariant \cite{ames2017control}, ensuring the safety of \eqref{eq:nom_sys}.

While CBFs provide a systematic way to synthesize safe controllers, an arbitrary safety function $h$ may not satisfy the CBF condition \eqref{eq:CBF}, especially with bounded inputs. In this case, there may exist states in $\Cs$ for which there is no admissible input satisfying the constraint \eqref{eq:CBF_constraint}. Backup CBFs \cite{gurriet_scalable_2020} address this issue by constructing a subset of $\Cs$ where safe inputs are guaranteed to exist and by generating corresponding safe controllers within the CBF framework. 

The backup CBF method leverages a backup controller and a backup set that typically represent safe but conservative behavior. Let $\Cb \subseteq \Cs$ be a backup set, defined by
\begin{equation} 
\label{eq:backup_set}
    \Cb \defeq \{\bs{x} \in \mathcal{X}: h_{\rm b}(\bs{x}) \geq 0\},
\end{equation}
with a continuously differentiable safety function ${h_{\rm b} \!:\! \mathcal{X} \!\to\! \R}$, and ${\ub \!:\! \mathcal{X} \!\to\! \mathcal{U}}$ be a backup controller that renders $\Cb$ forward invariant. The closed-loop dynamics under ${\ub}$ are
    ${\dot{\bs{x}}
    \!=\!
    \bs{f}(\bs{x})
    \!+\!
    \bs{g}(\bs{x})\ub(\bs{x})}$,
and the backup flow ${\bs{\varphi}_{\rm b}(\tau,\bs{x})}$ denotes the solution of this system starting from state $\bs{x}$, over time ${\tau \!\in\! [0,T]}$ with a backup horizon ${T \!>\! 0}$.

The key idea of bCBFs is to define the implicit safe set
\begin{equation} 
\label{eq:implicit_set}
    \mathcal{C}_{\rm I}
    \defeq
    \left\{
        \bs{x} \in \mathcal{X} : ~
        \begin{array}{@{}l@{}}
            h\big(\bs{\varphi}_{\rm b}(\tau,\bs{x})\big) \geq 0,
            \ \  \forall \tau \in [0,T],
            \\
            h_{\rm b}\big(\bs{\varphi}_{\rm b}(T,\bs{x})\big) \geq 0
        \end{array}
    \right\}.
\end{equation}
This represents the states from which the system can safely reach the backup set in time $T$ via the backup controller. The set $\mathcal{C}_{\rm I}$ is a subset of $\Cs$, and it is forward invariant under the backup controller. Therefore, by construction, safe inputs are guaranteed to exist in $\mathcal{C}_{\rm I}$. This construction is useful for input-constrained systems because feasibility can be inherited from the backup controller \cite{gurriet_scalable_2020}.

\section{Problem Formulation}
\label{sec:problem}
We extend bCBFs to systems with model uncertainty as
\begin{equation}
\label{eq:unc_sys_raw}
    \dot{\bs{x}}
    =
    \bs{f}(\bs{x})
    +
    \bs{g}(\bs{x})\bs{u}
    +
    \Delta\bs{f}(\bs{x})
    +
    \Delta\bs{g}(\bs{x})\bs{u},
\end{equation}
where ${\bs{f}:\mathcal{X}\to\R^n}$ and ${\bs{g}:\mathcal{X}\to\R^{n\times m}}$ are the known nominal components of the dynamics, while ${\Delta\bs{f}:\mathcal{X}\to\R^n}$ and ${\Delta\bs{g}:\mathcal{X}\to\R^{n\times m}}$ represent unknown structured uncertainty in the drift vector field and actuation matrix. We assume that the uncertainty admits the parameter affine representation
\begin{equation}
\label{eq:param_unc}
    \Delta\bs{f}(\bs{x})
    =
    \bs{F}(\bs{x})\bs{\theta}_{\rm f},
    \quad
    \Delta\bs{g}(\bs{x})\bs{u}
    =
    \sum\nolimits_{j=1}^{N_{\rm g}}
    \bs{G}_j(\bs{x}) ~\! \bs{u} ~\! \theta_{{\rm g},j},
\end{equation}
where ${\bs{F} \!:\! \mathcal{X} \! \to \! \R^{n\times N_{\rm f}}}$ and ${\bs{G}_j \!:\! \mathcal{X} \!\to\! \R^{n \times m}}$ are known terms, while ${\bs{\theta}_{\rm f} \in \R^{N_{\rm f}}}$ and ${\bs{\theta}_{\rm g} \in \R^{N_{\rm g}}}$ are unknown parameters.

Collecting the ${N \!=\! N_{\rm f} \!+\! N_{\rm g}}$ parameters in a vector ${\bs{\theta} \!\in\! \R^{N}}$ such that
   ${ \bs{\theta}
    \!\defeq\!
    \begin{bmatrix}
        \bs{\theta}_{\rm f}^{\top} &
        \bs{\theta}_{\rm g}^{\top}
    \end{bmatrix}^{\top}  \!=\! 
    \begin{bmatrix} \theta_1 & \theta_2 &  \ldots & \theta_{N}
    \end{bmatrix}^{\top}}$,
and by defining a regression matrix ${\boldsymbol{\phi}(\boldsymbol{x},\boldsymbol{u}) \in \mathbb{R}^{n \times N}}$ as
\begin{equation}
\label{eq:regressor_affine}
    \bs{\phi}(\bs{x},\bs{u})
    \defeq
    \begin{bmatrix}
        \bs{F}(\bs{x}) & \bs{G}_1(\bs{x})\bs{u} &
        \cdots &
        \bs{G}_{N_{\rm g}}(\bs{x})\bs{u}
    \end{bmatrix},
\end{equation}
system \eqref{eq:unc_sys_raw} can be written compactly as
\begin{equation}
\label{eq:unc_sys_affine}
    \dot{\bs{x}}
    =
    \bs{f}(\bs{x})
    +
    \bs{g}(\bs{x})\bs{u}
    +
    \bs{\phi}(\bs{x},\bs{u})\bs{\theta}.
\end{equation}

While $\bs{\theta}$ is unknown, it represents fixed coefficients that parameterize model uncertainty, a structure also common to adaptive and learning-based control. In robotic systems, these coefficients often correspond to physical quantities or model mismatches, such as damping coefficients, gravity mismatch, mass or inertia errors, for which conservative ranges can be obtained from system identification, nominal specifications, or physical limits. Accordingly, we assume that we know lower and upper bounds on each component of $\bs{\theta}$.
\begin{assumption} 
\label{assump:theta_initial_box}
There exist known constants ${\theta_i^{-}, \theta_i^{+} \!\in\! \R}$ for ${i \!\in\! \{1, \dots, N\}}$ such that ${\theta_i^{-} \!\leq\! \theta_i \!\leq\! \theta_i^{+}}$. Therefore, ${\bs{\theta} \!\in\! \Theta}$, where ${ \Theta
    \!\defeq\!
    [\theta_1^{-},\theta_1^{+}]
    \times
    \cdots
    \times
    [\theta_{N}^{-},\theta_{N}^{+}]}$.
\end{assumption}

Given Assumption~\ref{assump:theta_initial_box}, the unknown parameter $\bs{\theta}$ will be estimated, and the estimate will be used for safe control design. The proposed safety result does not require a specific adaptive law for estimation. Instead, it requires a continuously differentiable parameter estimate and certified component-wise bounds on the estimation error. We next define an estimator class that can provide the required estimate and bounds.

Let $\mathcal{H}_{t}$ denote the measured state and input data history available at time $t$, which may contain filtered regressions, integral regressions, or stored data. These data are used for calculating an estimate ${\thetahat(t) \in \R^{N}}$ of the parameter $\bs{\theta}$. We consider general parameter estimation dynamics of the form
\begin{equation}
\label{eq:theta_adaptation_law}
    \dot{\thetahat}
    =
    \bs{q}(\thetahat, \mathcal{H}_{t}, t),
\end{equation}
such that the estimator is initialized as ${\thetahat(0)\! \in \!\Theta}$ and satisfies
\begin{equation}
\label{eq:estimate_in_theta}
    \thetahat(t) \in \Theta,
    \quad \forall t \geq 0.
\end{equation}

\begin{definition}
\label{def:elementwise_estimator}
An estimator satisfying \eqref{eq:theta_adaptation_law} and  \eqref{eq:estimate_in_theta} is an \textit{element-wise certified adaptive estimator} if it provides a known continuously differentiable signal
${\bs{\rho}(t) \in \R_{\geq 0}^{N}}$ such that the estimation error $\thetatilde(t) \defeq \bs{\theta} - \thetahat(t)$ satisfies
\begin{equation}
\label{eq:rho_componentwise_bound}
    \!|\tilde{\theta}_i(t)|
    \!\leq\!
    \rho_i(t) \leq \theta_i^{+} - \theta_i^{-},
    \  \bs{\rho}(t) \!=\! \begin{bmatrix} \rho_1\!(t)  &  \!\!\!\!\!\!\ldots\!\!\!\!\! & \rho_N\!(t)
    \end{bmatrix}^{\top}\!\!, 
\end{equation}
for ${i \in \{1,\dots, N\}}$ and for all ${t \geq 0}$, independently of the applied controller. Equivalently, this can be expressed as
\begin{equation} 
\label{eq:error_set}
    \thetatilde(t)
    \in
    \Thetatilde
    \defeq
    \{
        \bs{\eta} \in \R^{N}:
        \bs{A}\bs{\eta}
        \leq
        \bs{b}(t)
    \},
\end{equation}
where $\Thetatilde$ is a polytope with ${\bs{A}
    \!\defeq\!
    \begin{bmatrix}
        \eye & \!\!\!-\eye
    \end{bmatrix}^\top}$ and
     ${\bs{b}(t)
    \!\defeq\!
    \begin{bmatrix}
        \bs{\rho}(t)^{\top} & \!\!\!\bs{\rho}(t)^{\top} 
    \end{bmatrix}^\top}$.
\end{definition}

\begin{example}[Dynamic regressor extension and mixing (DREM) parameter estimator \cite{aranovskiy2017performance, ortega2021new}]
\label{ex:drem}
Consider a linear regression ${y_{\rm r}(t) = \bs{m}_{\rm r}^{\top}(t)\bs{\theta}}$, where ${y_{\rm r}(t) \in \R}$ and ${\bs{m}_{\rm r}(t) \in \R^{N}}$ are known. The DREM estimator estimates a parameter in two steps. First, the extension step constructs known signals ${\bs{Y}_{\rm e}(t)\in\R^{N}}$ and ${\bs{M}_{\rm e}(t)\in\R^{N \times N}}$ satisfying ${\bs{Y}_{\rm e}(t) = \bs{M}_{\rm e}(t)\bs{\theta}}$. This extended regression can be obtained, for example, by stacking the original regression with filtered regressions generated by stable linear operators \cite{aranovskiy2017performance}. Second, the mixing step multiplies the extended regression by the adjugate of $\bs{M}_{\rm e}(t)$ as ${\bs{Y}_{\rm D}(t) \defeq \operatorname{adj}(\bs{M}_{\rm e}(t))\bs{Y}_{\rm e}(t)}$. By defining ${\chi(t) \defeq \det(\bs{M}_{\rm e}(t))}$, it can be shown that ${\bs{Y}_{\rm D}(t) = \chi(t)\bs{\theta}}$ holds because ${\operatorname{adj}(\bs{M}_{\rm e})\bs{M}_{\rm e} = \det(\bs{M}_{\rm e})\bs{I}_{N}}$. Then, each parameter is estimated with the update law
\begin{equation}
\label{eq:DREM_law}
\dot{\hat{\theta}}_i(t)
    =
    \gamma_i\chi(t)
    \big(
        Y_{{\rm D},i}(t)
        -
        \chi(t)\hat{\theta}_i(t)
    \big),
    \ \
    \gamma_i>0, 
\end{equation}
which yields the estimation error dynamics
\begin{equation}
\label{eq:DREM_err_dyn}
\dot{\tilde{\theta}}_i(t)
    = -\gamma_i\chi^2(t)\tilde{\theta}_i(t). 
\end{equation}
\begin{lemma}
\label{lemma:ro_bound}
Consider the uncertain system \eqref{eq:unc_sys_affine} and update law \eqref{eq:DREM_law} for the estimation of the parameter $\bs{\theta}$. Under Assumption~\ref{assump:theta_initial_box} and the initialization ${\thetahat(0) \in \Theta}$, the DREM estimator in Example~\ref{ex:drem} satisfies \eqref{eq:estimate_in_theta} and \eqref{eq:rho_componentwise_bound}.  
\end{lemma}
\begin{proof}
Let ${\nu_i(t) \!\defeq\! {\rm e}^{ -\gamma_i \int_0^t \chi^2(s) \!~d s }}$. Since ${\gamma_i \!>\! 0}$, we have ${0 \!<\! \nu_i(t) \!\leq\! 1}$ for all ${t \!\geq\! 0}$. The solution of \eqref{eq:DREM_err_dyn} with the initial condition $\tilde{\theta}_i(0)$ is given by ${\tilde{\theta}_i(t) \!=\! \nu_i(t) \tilde{\theta}_i(0)}$. Thus, $|\tilde{\theta}_i(t)|$ satisfies the error bound \eqref{eq:rho_componentwise_bound} with ${\rho_i(t)
    \!=\!
    \nu_i(t)
    \rho_i(0)}$, if ${\rho_i(0) \!\geq\! |\tilde{\theta}_i(0)|}$. For $\rho_i(0)$, Assumption~\ref{assump:theta_initial_box} and ${\thetahat(0) \!\in\! \Theta}$ allow the choice
    ${\rho_i(0)
    \!=\!
    \max\{
        \hat{\theta}_i(0) \!- \theta_i^{-},~
        \theta_i^{+} \! -\hat{\theta}_i(0)
    \}}$.
Note that the least conservative bound is obtained for $\rho_i(t)$ if the initial value of $\hat{\theta}_i$ is the midpoint of the interval, given by ${\hat{\theta}_i(0) \!=\! (\theta_i^{-} \!+\! \theta_i^{+})/2}$. Finally, we prove that ${\thetahat(t) \!\in\! \Theta}$ for all ${t \!\geq\! 0}$. Using ${\tilde{\theta}_i(t) \!=\! \theta_i - \hat{\theta}_i(t)}$ gives 
${\hat{\theta}_i(t) \!=\! \big(1 \!- \nu_i(t) \big) \theta_i \!+\!
    \nu_i(t) \hat{\theta}_i(0)}$. By Assumption~\ref{assump:theta_initial_box}, ${\theta_i \!\in\! [\theta_i^{-}, \theta_i^{+}]}$, and by
${\thetahat(0) \!\in\! \Theta}$, ${\hat{\theta}_i(0) \!\in\! [\theta_i^{-}, \theta_i^{+}]}$. Hence, ${\hat{\theta}_i(t) \!\in\! [\theta_i^{-}, \theta_i^{+}]}$ for all ${t \!\geq\! 0}$ and ${i \!\in\! \{1, \dots, N\}}$, implying ${\thetahat(t) \!\in\! \Theta, \!\ \forall t \!\geq\! 0}$. Thus, the estimator satisfies Def.~\ref{def:elementwise_estimator}. 
\end{proof}
\end{example}

We remark that DREM is only one example realization of the estimator class. For instance, recursive least-squares estimators with zonotopic uncertainty propagation also provide an adaptation law with a set-valued uncertainty certificate \cite{cohen2024}. More generally, any adaptive, composite learning, or set-membership estimator can be used if it provides the adaptation law in \eqref{eq:theta_adaptation_law} and an element-wise bound  \cite{cohen2024,tarun2025}. 

We are now in a position to state the considered problem, which will be addressed using robust adaptive bCBFs. 
\begin{problem}
Given the uncertain system \eqref{eq:unc_sys_affine}, the safe set $\Cs$ in \eqref{eq:safeset}, and an element-wise certified adaptive estimator \eqref{eq:theta_adaptation_law}, design a control law ${\bs{u}=\bs{k}(t,\bs{x},\thetahat) \in \mathcal{U}}$ that renders a subset of $\Cs$
forward invariant, ensuring the safety of \eqref{eq:unc_sys_affine}.
\end{problem}

\section{Main Results}
\label{sec:main}
We now develop the proposed {\em robust adaptive bCBF method} in four steps.
First, we propagate the backup flow using the current parameter estimate.
Second, we bound the mismatch between the true and estimated backup flows using the certified component-wise parameter error bounds.
Third, we use this bound to define a tightened inner approximation of the implicit safe set.
Fourth, we derive duality-based robust safety conditions that preserve a QP implementation.

Similar to the standard bCBF method outlined in Section~\ref{sec:pre}, the proposed approach leverages a backup controller ${\ub:\mathcal{X}\to\mathcal{U}}$ and a backup set ${\Cb \subseteq \Cs}$ given in \eqref{eq:backup_set}. This time, we assume that $\ub$ is a robust backup controller.  

\begin{assumption} 
\label{assump:robust_backup_main}
The backup controller $\ub$ renders $\Cb$ forward invariant for \eqref{eq:unc_sys_affine} for any $\bs{\theta} \in \Theta$.
\end{assumption}

In practice, robust backup sets can often be obtained by robustifying Lyapunov level sets around a stabilizable equilibrium point and verifying their invariance under the corresponding stabilizing feedback control law \cite{van2024disturbance}. In Section~\ref{sec:sim}, we utilize this approach to synthesize a robust backup controller and a robust backup set for a quadrotor model.

We introduce the closed-loop dynamics of the system \eqref{eq:unc_sys_affine} under the robust backup controller as
\begin{equation}
\label{eq:backup_dyn_param}
    \dot{\bs{x}} =
    \bs{f}_{\rm b}(\bs{x},\bs{\theta})
    \defeq
    \bs{f}(\bs{x})
    +
    \bs{g}(\bs{x})\ub(\bs{x})
    +
    \bs{\psi}_{\rm b}(\bs{x})\bs{\theta},
\end{equation}
where ${\bs{\psi}_{\rm b}(\boldsymbol{x}) \in \mathbb{R}^{n \times N}}$ indicates the regression matrix in \eqref{eq:regressor_affine} evaluated along the backup controller:
\begin{equation} 
\label{eq:regression_backup}
    \bs{\psi}_{\rm b}(\bs{x})
    \defeq
    \bs{\phi}(\bs{x},\ub(\bs{x})).
\end{equation}
We assume that $\bs{f}_{\rm b}$ and $\bs{\psi}_{\rm b}$ are continuously differentiable in their arguments. Then, the backup flow satisfies
\begin{equation}
\label{eq:true_backup_flow_main}
    \tfrac{\partial}{\partial \tau}\varphibtrue{\tau}{\bs{x}}
    =
    \bs{f}_{\rm b}\big(\varphibtrue{\tau}{\bs{x}},\bs{\theta}\big),
    \ \ 
    \varphibtrue{0}{\bs{x}}
    =
    \bs{x}.
\end{equation}

Ideally, one would use the backup flow to enforce the forward invariance of the implicit safe set $\mathcal{C}_{\rm I}$ defined in \eqref{eq:implicit_set}. However, since the parameter $\bs{\theta}$ is unknown, the backup flow $\bs{\varphi}_{\rm b}$ and the implicit set $\mathcal{C}_{\rm I}$ may not be known. Therefore, we establish a known subset of $\mathcal{C}_{\rm I}$ and enforce the forward invariance thereof. To this end, we propagate an estimated backup flow $\widehat{\bs{\varphi}}_{\rm b}$ with $\thetahat(t)$ and $\rho_i(t)$ fixed over the backup horizon:
\begin{equation}
\label{eq:estimated_backup_flow_main}
    \! \! \! \tfrac{\partial}{\partial \tau}\varphibhat{\tau}{\bs{x}}
    \!=\!
    \bs{f}_{\rm b}\big(\varphibhat{\tau}{\bs{x}},\thetahat(t)\big),
    \ 
    \varphibhat{0}{\bs{x}}
    \!=\!
    \bs{x}.
\end{equation}

While the estimated flow $\widehat{\bs{\varphi}}_{\rm b}$ is known, it is not sufficient for ensuring safety by itself, since the true parameter $\bs{\theta}$ can differ from the estimated parameter $\thetahat$. Therefore, we bound the difference between the true and estimated backup flows:
\begin{equation} 
\label{eq:flow_bound}
    \big \| {
    \varphibtrue{\tau}{\bs{x}} - \varphibhat{\tau}{\bs{x}}
    } \big \|
   \leq \delta_{\varphi}(\tau,\bs{x},t),
\end{equation}
and use these bounds to establish robust safety. The following lemma provides one possible flow error bound $\delta_{\varphi}$.

\begin{lemma}
\label{lemma:delta_bound_main}
Assume that $\bs{f}_{\rm b}(\bs{x},\bs{\theta})$ in \eqref{eq:backup_dyn_param} is uniformly Lipschitz on $\mathcal{X}$ over all $\bs{\theta}\in\Theta$ with Lipschitz constant $L_{\rm b}$. Then, \eqref{eq:flow_bound} holds for all $\tau\in[0, T]$, $\bs{x}\in\mathcal{X}$, and ${t \geq 0}$ with
\begin{equation} 
\label{eq:delta_Gronwall}
\begin{aligned}
    \delta_{\varphi}(\tau,\bs{x},t) & \defeq
    \int_0^{\tau} {\rm e}^{L_{\rm b}(\tau-s)}
    d_{\rm b}(s,\bs{x},t) ~\!d s, \\
    d_{\rm b}(\tau,\bs{x},t) & \defeq
    \sum\nolimits_{i=1}^{N}
    \norm{
        \bs{\psi}_{{\rm b},i}
        \big(\varphibhat{\tau}{\bs{x}}\big)
    }_{\rm s} 
    \rho_i(t),
\end{aligned}
\end{equation}
where $\bs{\psi}_{{\rm b},i}$ is the $i$th column of $\bs{\psi}_{\rm b}$ in \eqref{eq:regression_backup}, $\norm{\cdot}_{\rm s}$ represents a smooth upper bound of the Euclidean norm\footnote{In this work, we use ${\norm{\bs{z}}_{\rm s} = \sqrt{\norm{\bs{z}}^2 + \sigma^2}}$ with ${\sigma > 0}$, which is continuously differentiable and satisfies ${\norm{\bs{z}}_{\rm s} \geq \norm{\bs{z}}}$ for all ${\bs{z} \in \R^n}$.}, and $\rho_i(t)$ is the element-wise parameter estimation error bound in \eqref{eq:rho_componentwise_bound}.
\end{lemma}
\begin{proof}
First, we derive the dynamics of the error term
${\bs{e}(\tau,\bs{x},t) \defeq \varphibtrue{\tau}{\bs{x}} - \varphibhat{\tau}{\bs{x}}}$.
Using \eqref{eq:true_backup_flow_main} and \eqref{eq:estimated_backup_flow_main}, and adding and subtracting $\bs{f}_{\rm b}(\varphibhat{\tau}{\bs{x}},\bs{\theta})$, we have
\begin{equation*}
\begin{aligned}
    \tfrac{\partial}{\partial \tau} \bs{e}(\tau,\bs{x},t)
    &=
    \bs{f}_{\rm b}
    \big(\varphibtrue{\tau}{\bs{x}},\bs{\theta}\big)
    -
    \bs{f}_{\rm b}
    \big(\varphibhat{\tau}{\bs{x}},\bs{\theta}\big)
    \\
    & \quad +
    \bs{\psi}_{\rm b}
    \big(\varphibhat{\tau}{\bs{x}}\big)
    \thetatilde(t).
\end{aligned}
\end{equation*}
The last term can be bounded using $d_{\rm b}$ as
\begin{equation*}
\begin{aligned}
    &
    \norm{
    \bs{\psi}_{\rm b}
    \big(\varphibhat{\tau}{\bs{x}}\big) 
    \thetatilde(t)}
     =
    \norm{
    \sum\nolimits_{i=1}^{N}
    \bs{\psi}_{{\rm b},i}
    \big(\varphibhat{\tau}{\bs{x}}\big)
    \tilde{\theta}_i(t)}
    \\
    & ~~~~~~~~~~~ \leq
    \sum\nolimits_{i=1}^{N}
    \norm{
        \bs{\psi}_{{\rm b},i}
        \big(\varphibhat{\tau}{\bs{x}}\big)
    }_{\rm s} 
    \rho_i(t) = d_{\rm b}(\tau,\bs{x},t).
\end{aligned}
\end{equation*}

Next, using the integral form of the error dynamics and ${\bs{e}(0,\bs{x},t) = \bzero}$, we get
\begin{equation*}
\begin{aligned}
    \norm{\bs{e}(\tau,\bs{x},t)}
    & \!\leq \! \!
    \int_0^{\tau} \!\!
    \norm{
    \bs{f}_{\rm b}
    \big(\varphibtrue{s}{\bs{x}},\bs{\theta}\big)
    \!-\!
    \bs{f}_{\rm b}
    \big(\varphibhat{s}{\bs{x}},\bs{\theta}\big)
    } d s
    \\
    & \quad +
    \int_0^{\tau}
    \norm{
    \bs{\psi}_{\rm b}
    \big(\varphibhat{s}{\bs{x}}\big)
    \thetatilde(t)
    } d s,
    \\
    & \leq
    \int_0^{\tau}
    \big(
    L_{\rm b}
    \norm{\bs{e}(s,\bs{x},t)}
    +
    d_{\rm b}(s,\bs{x},t)
    \big) d s,
\end{aligned}
\end{equation*}
where the first inequality follows from the triangle inequality and the second inequality holds because of the uniform Lipschitz property of $\bs{f}_{\rm b}$ and the definition of $d_{\rm b}$. Thus, the bound in \eqref{eq:delta_Gronwall} follows from the Gronwall--Bellman inequality.
\end{proof}

We note that $\delta_{\varphi}$ in \eqref{eq:delta_Gronwall} is only one possible flow bound. This bound can lead to conservative tightening for long prediction horizons or large parameter error bounds. When more information about the closed-loop system may be leveraged (e.g., contraction, one-sided Lipschitzness \cite{contraction_bullo}, or linearity), this bound can be made tighter \cite{van2025uncertainty}.

Having established the flow error bound in \eqref{eq:flow_bound}, we now introduce tightening terms that allow us to make constraints on the estimated backup flow sufficient for satisfying constraints on the true backup flow. These tightening terms are denoted as $\epsilon$ and $\epsilon_{\rm b}$, and they are constructed such that
\begin{equation} 
\label{eq:epsilon}
\begin{aligned}
   \epsilon(\tau, \bs{x}, t) & \geq
   h(\varphibhat{\tau}{\bs{x}}) - h(\varphibtrue{\tau}{\bs{x}}), \\
   \epsilon_{\rm b}(\bs{x},t) & \geq
   h_{\rm b}(\varphibhat{T}{\bs{x}}) - h_{\rm b}(\varphibtrue{T}{\bs{x}}),
\end{aligned}
\end{equation}
hold for all ${\tau \in [0, T]}$, ${\bs{x} \in \mathcal{X}}$, and ${t \geq 0}$, providing bounds on the effect of the parameter uncertainty on safety. For example, by the Lipschitz continuity of $h$ and $h_{\rm b}$, we may define continuously differentiable tightening functions as
\begin{equation}
\label{eq:eps_b_choice_main}
    \epsilon(\tau, \bs{x}, t)
    \defeq 
    L_h \delta_{\varphi}(\tau,\bs{x},t), \ 
    \epsilon_{\rm b}(\bs{x}, t)
    \defeq  
    L_{h_{\rm b}}\delta_{\varphi}(T,\bs{x}, t),
\end{equation}
where $L_h$ and $L_{h_{\rm b}}$ are the Lipschitz constants of $h$ and $h_{\rm b}$. For convex or quadratic barriers, even tighter $\epsilon$ and $\epsilon_{\rm b}$ terms may be obtained \cite[Remark 2]{van2026output}. 

To enhance robustness against parameter estimation errors, similar to the uncertainty estimator-based bCBF method proposed in \cite{van2025uncertainty}, we define a time-varying inner approximation $\Cdhat$ of the true implicit safe set $\mathcal{C}_{\rm I}$ in \eqref{eq:implicit_set}:
\begin{equation} 
\label{eq:implicit_set_approx}
    \Cdhat
    \defeq
    \left\{
    \bs{x} \in \mathcal{X} : ~
    \begin{array}{@{}l@{}}
    \bar h(\tau, \bs{x}, t) \geq 0,
    \ \  \forall \tau \in [0, T],
    \\
    \bar h_{\rm b}(\bs{x},t) \geq 0
    \end{array}
    \right\},
\end{equation}
for all ${t \geq 0}$, where the tightening terms are incorporated in
\begin{equation} 
\label{eq:barriers_subset}
\begin{aligned}
    \bar h(\tau, \bs{x}, t) & \defeq h(\varphibhat{\tau}{\bs{x}})  \!-\! \epsilon(\tau, \bs{x}, t), \\
    \bar h_{\rm b}(\bs{x},t) & \defeq h_{\rm b}(\varphibhat{T}{\bs{x}})  -  \epsilon_{\rm b}(\bs{x}, t).
\end{aligned}
\end{equation}
We remark that \eqref{eq:epsilon} implies $\Cdhat \subseteq\mathcal{C}_{\rm I}$. 

Having constructed a known subset of $\mathcal{C}_{\rm I}$, we now derive sufficient conditions for the forward invariance of $\Cdhat$. We impose that the control input $\bs{u}$ must satisfy
\begin{equation} 
\label{eq:sufficient_cond}
\begin{aligned}
    \dot{\bar h}(\tau,\bs{x},t,\bs{u})
    &\geq
    -\alpha\big(\bar h(\tau,\bs{x},t)\big),
    \ \
    \forall ~\! \tau\in[0,T],
    \\
    \dot{\bar h}_{\rm b}(\bs{x},t,\bs{u})
    &\geq
    -\alpha_{\rm b}\big(\bar h_{\rm b}(\bs{x},t)\big),
\end{aligned}
\end{equation}
for all ${\bs{x} \in \Cdhat}$ and ${t \geq 0}$. By calculating the time derivatives via the chain rule, these inequalities can be written as
\begin{equation} 
\label{eq:uncertain_bCBF_constraint}
\begin{aligned}
    a(\tau,\bs{x},\bs{u},t)
    +
    \bs{c}(\tau,\bs{x},\bs{u},t)^{\top} \thetatilde(t)
    & \geq 0,
    \\
    a_{\rm b}(\bs{x},\bs{u},t)
    +
    \bs{c}_{\rm b}(\bs{x},\bs{u},t)^{\top} \thetatilde(t)
    & \geq 0,
\end{aligned}
\end{equation}
where 
\begin{equation*}
\begin{aligned}
    a(\tau,\bs{x},\bs{u},t) & \defeq
    \derp{\bar{h}}{\bs{x}}(\tau,\bs{x},t)
    \big( \bs{f}(\bs{x}) + \bs{g}(\bs{x})\bs{u} + \bs{\phi}(\bs{x},\bs{u})\thetahat(t) \big)
    \\
    & \quad
    + \derp{\bar{h}}{t}(\tau,\bs{x},t)
    + \alpha\big(\bar h(\tau,\bs{x},t)\big),
    \\
    \bs{c}(\tau,\bs{x},\bs{u},t)^{\top} & \defeq
    \derp{\bar{h}}{\bs{x}}(\tau,\bs{x},t)
    \bs{\phi}(\bs{x},\bs{u}),
    \\
    a_{\rm b}(\bs{x},\bs{u},t) & \defeq
    \derp{\bar{h}_{\rm b}}{\bs{x}}(\bs{x},t)
    \big( \bs{f}(\bs{x}) + \bs{g}(\bs{x})\bs{u} + \bs{\phi}(\bs{x},\bs{u})\thetahat(t) \big)
    \\
    & \quad
    + \derp{\bar{h}_{\rm b}}{t}(\bs{x},t)
    + \alpha_{\rm b}\big(\bar h_{\rm b}(\bs{x},t)\big),
    \\
    \bs{c}_{\rm b}(\bs{x},\bs{u},t)^{\top} & \defeq
    \derp{\bar{h}_{\rm b}}{\bs{x}}(\bs{x},t)
    \bs{\phi}(\bs{x},\bs{u}),
\end{aligned}
\end{equation*}
and these expressions are derived in the Appendix. Note that while the $a$, $a_{\rm b}$, $\bs{c}$, and $\bs{c}_{\rm b}$ coefficients are known, the parameter estimation error $\thetatilde(t)$ is unknown, and hence \eqref{eq:uncertain_bCBF_constraint} cannot be enforced directly.

Instead, we apply the robust adaptive bCBF conditions
\begin{align}
    a(\tau,\bs{x},\bs{u},t)
    +
    \inf_{\bs{\eta}\in\Thetatilde}
    \bs{c}(\tau,\bs{x},\bs{u},t)^{\top}\bs{\eta}
    & \geq
    0,
    \label{eq:robust_h_primal_main}
    \\
    a_{\rm b}(\bs{x},\bs{u},t)
    +
    \inf_{\bs{\eta}\in\Thetatilde}
    \bs{c}_{\rm b}(\bs{x},\bs{u},t)^{\top}\bs{\eta}
    & \geq
    0.
    \label{eq:robust_hb_primal_main}
\end{align} 
These inequalities enforce the bCBF conditions for all admissible parameter estimation errors inside the known error set $\Thetatilde$ given by \eqref{eq:error_set}. The inner minimizations are linear programs (LPs) over the component-wise estimation error set. Using LP duality \cite{cohen2022robust}, the robust constraints are enforced by the existence of dual variables $\bs{\mu}_{\tau},\bs{\mu}_{\rm b} \in \R^{2 N}$ such that
\begin{align}
    & a(\tau,\bs{x},\bs{u},t)
    +
    \bs{b}(t)^{\top}\bs{\mu}_{\tau}
     \geq
    0,
    \label{eq:main_h_dual_main}
    \\
& \bs{\mu}_{\tau}^{\top}\bs{A}
     =
    \bs{c}(\tau,\bs{x},\bs{u},t)^{\top}, 
    ~~~~\bs{\mu}_{\tau}
     \leq
    \bzero,
    \label{eq:main_eq_h_main}
    \\
    & a_{\rm b}(\bs{x},\bs{u},t)
    +
    \bs{b}(t)^{\top}\bs{\mu}_{\rm b}
     \geq
    0,
    \label{eq:main_hb_dual_main}
    \\
    & \bs{\mu}_{\rm b}^{\top}\bs{A}
    =
    \bs{c}_{\rm b}(\bs{x},\bs{u},t)^{\top},
    ~~~~~ \bs{\mu}_{\rm b}
     \leq
    \bzero,
    \label{eq:main_eq_hb_main}
\end{align}
with $\bs{A}$ and $\bs{b}(t)$ from \eqref{eq:error_set}. Note that the functions $a$, $a_{\rm b}$, $\bs{c}$, and $\bs{c}_{\rm b}$ are affine in $\bs{u}$ as $\bs{\phi}(\bs{x},\bs{u})$ is affine in $\bs{u}$. Hence, \eqref{eq:main_h_dual_main}--\eqref{eq:main_eq_hb_main} are affine in the decision variables $\bs{u}$, $\bs{\mu}_{\tau}$, and $\bs{\mu}_{\rm b}$.
\begin{theorem}
\label{thm:main_adaptive_bcbf_main}
Consider the uncertain system \eqref{eq:unc_sys_affine}, an estimator in Definition~\ref{def:elementwise_estimator}, the safe set $\Cs$ in \eqref{eq:safeset}, the backup controller $\ub$, the backup set $\Cb$ in \eqref{eq:backup_set}, and the time-varying set $\Cdhat$ in \eqref{eq:implicit_set_approx}. Suppose Assumptions~\ref{assump:theta_initial_box} and~\ref{assump:robust_backup_main} hold. Then, any locally Lipschitz controller satisfying \eqref{eq:main_h_dual_main}--\eqref{eq:main_eq_hb_main} for all ${\tau\in[0,T]}$ renders $\Cdhat$ forward invariant for \eqref{eq:unc_sys_affine}. Consequently, every closed-loop trajectory starting in $\widehat{\mathcal{C}}_{\rm I}(0)$ remains in $\Cs$:
${\bx(0) \in \widehat{\mathcal{C}}_{\rm I}(0) \implies \bx(t) \in \Cs}$, ${\forall t \geq 0}$.
\end{theorem}
\begin{proof}
By LP weak duality, the dual inequality and equality conditions in \eqref{eq:main_h_dual_main}--\eqref{eq:main_eq_hb_main} imply the primal robust inequalities \eqref{eq:robust_h_primal_main} and \eqref{eq:robust_hb_primal_main}. Since $\thetatilde(t)\in\Thetatilde$, these primal inequalities imply \eqref{eq:sufficient_cond}. Consequently, the comparison lemma gives the forward invariance of $\Cdhat$. From the tightening construction \eqref{eq:epsilon}, we have $\Cdhat \subseteq \mathcal{C}_{\rm I}$. Since ${\tau=0}$ is included in the definition of $\mathcal{C}_{\rm I}$, we also have ${\mathcal{C}_{\rm I}\subseteq\Cs}$. Therefore, every closed-loop trajectory starting in $\widehat{\mathcal{C}}_{\rm I}(0)$ remains in $\Cs$.
\end{proof}

Having established the robust safety constraints in \eqref{eq:main_h_dual_main}--\eqref{eq:main_eq_hb_main}, we now introduce an optimization problem, in the form of a QP, that enables the synthesis of safe controllers. Let the backup horizon be discretized by ${\tau_j \!=\! j ~\! T/N_T}$ for ${j \!\in\! \{0,\dots,N_T\}}$, and define the decision variable vector as
    ${\bs{w}
    \!\defeq\!
    \big(
    \bs{u},
    \{\bs{\mu}_{\tau_j}\}_{j=0}^{N_T},
    \bs{\mu}_{\rm b}
    \big)}$
that contains ${N_w \!\defeq\! m + (N_T+2) 2 N}$ elements. With $\bs{w}$, we formulate an adaptive safety filter that minimally modifies a potentially unsafe primary controller ${\up: \mathcal{X} \to \mathcal{U}}$ into a safe controller:
\begin{align}
    \bs{k}^{\star}(t, \bs{x}, \thetahat)
    =
    \argmin_{\bs{w} \in \R^{N_w}} \ \
    &
    \norm{\bs{u}-\up(\bs{x})}^2
    \label{eq:adaptive_qp_main}
    \\
    \subjectto \ \
    &
    \eqref{eq:main_h_dual_main},\eqref{eq:main_eq_h_main}
    \text{ at } \tau_j, \ j \in \{0,\dots,N_T\},
    \nonumber\\
    &
    \eqref{eq:main_hb_dual_main}, \eqref{eq:main_eq_hb_main}, ~ \bs{u} \in \mathcal{U}. 
    \nonumber
\end{align}

\begin{remark}
Theorem~\ref{thm:main_adaptive_bcbf_main} uses continuous time safety conditions over ${\tau\in[0,T]}$, while the QP in \eqref{eq:adaptive_qp_main} enforces discretized constraints on a grid using $\tau_j$. As in standard bCBF implementations \cite{gurriet_scalable_2020}, one can add inter-sample margins or choose a sufficiently fine grid to account for the discretization error.
\end{remark}

\begin{remark}
The proposed method with the QP in \eqref{eq:adaptive_qp_main} ensures safety without requiring the notoriously restrictive persistent excitation condition \cite{isaly2021adaptive} on the regression matrix. Furthermore, as the certified bounds $\rho_i(t)$ decrease, the flow error bound $\delta_{\varphi}$ and the tightening terms $\epsilon,\epsilon_{\rm b}$ shrink; therefore, the inner approximation $\Cdhat$ approaches $\mathcal{C}_{\rm I}$ as ${\rho_i(t) \!\to\! 0}$. 
\end{remark}

\begin{remark}
Although ${\Cdhat \subseteq \mathcal{C}_{\rm I}}$, forward invariance of the set ${\mathcal{C}_{\rm I}}$ under $\ub$ need not extend to its subset $\Cdhat$ under the same controller. Therefore, the set $\Cdhat$ is not necessarily controlled invariant by construction and the QP in \eqref{eq:adaptive_qp_main} can be infeasible for some ${t \!\geq\! 0}$ because of the input bounds or the emptiness of the set $\Cdhat$. Note that the non-emptiness of $\Cdhat$ depends on the size of the tightening terms $\epsilon,\epsilon_{\rm b}$, which are largest at ${t \!=\! 0}$ and shrink as the certified bounds $\rho_i(t)$ decrease. In practice, $\Cdhat$ can be made non-empty by choosing a sufficiently short backup horizon $T$ or a sufficiently tight initial parameter box $\Theta$, and becomes progressively less conservative as the estimator refines the parameter estimate. To ensure robust safety in the case of potential infeasibility of the QP, the robust backup controller $\ub$ can be utilized. By Assumption~\ref{assump:robust_backup_main}, $\ub$ renders $\mathcal{C}_{\rm I}$ forward invariant along \eqref{eq:unc_sys_affine}, and satisfies ${\ub(\bs{x}) \in \mathcal{U}}$ for all ${\bs{x} \in \mathcal{C}_{\rm I}}$. Therefore, if \eqref{eq:adaptive_qp_main} becomes infeasible, one can switch to $\ub$ to keep the state in ${\mathcal{C}_{\rm I} \!\subseteq\! \Cs}$ while satisfying the input constraints. A smooth switching between $\bs{k}^{\star}$ and $\ub$ can also be used, as in \cite{rabiee2025soft}.
\end{remark}

\section{Planar Quadrotor Example}
\label{sec:sim}
We illustrate the proposed control method using a planar quadrotor example with parametric uncertainty in both the drift and input matrices\footnote{\label{footnote:code} The code and additional details of our simulations can be found at \\ \url{https://github.com/ersindas/abCBFs}}. The state and input vectors are
    $\bs{x}
    \!\defeq\!
    \begin{bmatrix}
        p_x & p_z & \vartheta & v_x & v_z & \omega
    \end{bmatrix}^{\top}
    \!\!\in\!
    \mathcal{X},
    ~~
    \!\bs{u}\!
    \defeq
    \begin{bmatrix}
        F & M
    \end{bmatrix}^{\top}
    \!\!\in\!
    \mathcal{U}$,
where $\mathcal{X} \subset \mathbb{R}^{2}\times \mathbb{S}^{1}\times \mathbb{R}^{3}$, $p_x$ and $p_z$ denote the horizontal and vertical quadrotor positions, $\vartheta$ is its pitch angle, and $(v_x, v_z,\omega)$ are the translational and angular velocities. The inputs are the thrust force $F$ and pitch moment $M$, with $\mathcal{U} \defeq [0,F_{\max}] \times [-M_{\max},M_{\max}] \subset \mathbb{R}^{2}$. The model is 
\begin{equation} 
\label{eq:uncertain_drone_dyn}
\begin{aligned}
\dot{p}_x & = v_x, \quad
\dot{p}_z = v_z, \quad
\dot{\vartheta} = \omega, \quad \dot{\omega} = \delta_{\ell} F -\tfrac{1}{J} M,
\\
\dot{v}_x &= -c_x v_x + \tfrac{\sin\vartheta}{m} F, \quad
\dot{v}_z = -c_z v_z - g + \tfrac{\cos\vartheta}{m} F, 
\end{aligned}
\end{equation}
where $c_x$ and $c_z$ are damping coefficients, $g$ is the gravitational acceleration, $m$ is the mass of the quadrotor, $J$ is its mass moment of inertia, and  $\delta_{\ell}$ captures an unknown pitch moment generated by thrust. We assume that each of these six parameters is uncertain, and the nominal values used for control design are $0$, $0$, $g_0$, $m_0$, $J_0$, and $0$, respectively. 

\begin{figure}[t]
    \centering
     \vspace{0.2cm}
\includegraphics[width=0.95\linewidth]{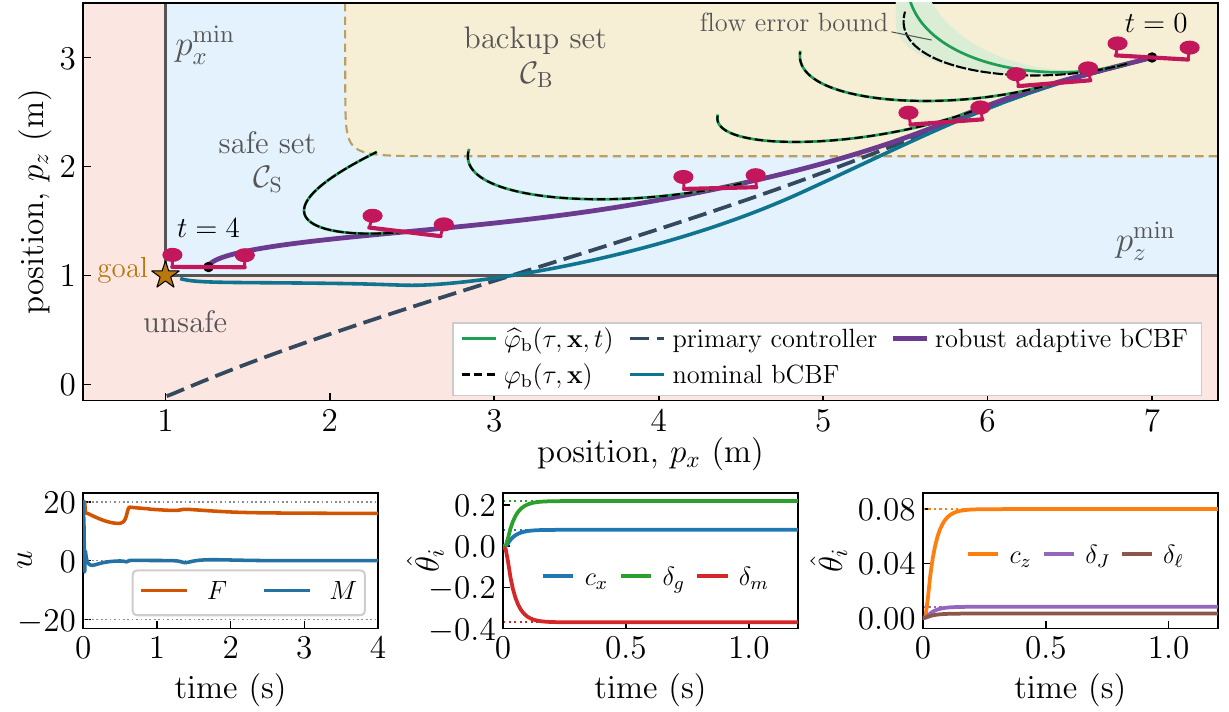}
    \vspace{-2mm}
    \caption{Simulation of the uncertain quadrotor model. \textbf{(Top)} The trajectories in the $p_x$--$p_z$ plane. The proposed robust adaptive bCBF-QP \eqref{eq:adaptive_qp_main} keeps the quadrotor inside the safe set $\mathcal{C}_{\rm S}$ under parametric uncertainty and bounded inputs, while the primary controller and the nominal bCBF violate safety. \textbf{(Bottom)} The inputs satisfy the input constraints (left). The six parameter estimates $\hat{\theta}_i$ (middle and right). Estimation error bounds shrink over time.}
   \label{fig:results}
   \vspace{-4mm}
\end{figure}

Accordingly, we introduce the unknown parameter vector
    $\bs{\theta}
    \defeq
    \begin{bmatrix}
        c_x & c_z & \delta_g & \delta_m & \delta_J & \delta_{\ell}
    \end{bmatrix}^{\top}
    \in
    \Theta
    \subset
    \mathbb{R}^{6}$,
with
${\delta_g \defeq g-g_0}$,
${\delta_m \defeq \frac{1}{m} - \frac{1}{m_0}}$,
${\delta_J \defeq \frac{1}{J} - \frac{1}{J_0}}$, and we separate the nominal model and the uncertainty as
\begin{equation*}
    \begingroup
    \setlength{\arraycolsep}{1.5pt}
    \dot{\bs{x}} 
    \!=\!\!\!
    \underbrace{{ 
    \begin{bmatrix}
        \!v_x\! \\
        \!v_z\! \\
        \!\omega\! \\
        \!0\! \\
        \!-g_0\! \\
        \!0\!
    \end{bmatrix}}}_{\bs{f}(\bs{x})}
    \!\!+\!\!
    \underbrace{{ 
    \begin{bmatrix}
        \!0 & 0 \\
        \!0 & 0 \\
        \!0 & 0 \\
        \!\frac{\sin\vartheta}{m_0} & 0 \\
        \!\frac{\cos\vartheta}{m_0} & 0 \\
        \!0 & -\frac{1}{J_0}
    \end{bmatrix}}}_{\bs{g}(\bs{x})}
    \!\!\bs{u}
    +\!
    \underbrace{{ 
    \begin{bmatrix}
        \!0 & 0 & 0 & 0 & 0 & 0 \\
        \!0 & 0 & 0 & 0 & 0 & 0 \\
        \!0 & 0 & 0 & 0 & 0 & 0 \\
        \!-v_x & 0 & 0 & F\! \sin \!\vartheta & 0 & 0 \\
        \!0 & -v_z & -1 & F\! \cos \!\vartheta & 0 & 0 \\
        \!0 & 0 & 0 & 0 & -M & F
    \end{bmatrix}}}_{\bs{\phi}(\bs{x},\bs{u})}
    \!\!\bs{\theta}.
    \endgroup
\end{equation*}
We remark that $c_x$, $c_z$, and $\delta_g$ introduce drift uncertainty, while the uncertainties $\delta_m$, $\delta_J$, and $\delta_{\ell}$ enter through the actuation matrix and multiply the control input, which is addressed by the duality-based robustification in Section~\ref{sec:main}. Note that the  regression matrix $\bs{\phi}(\bs{x},\bs{u})$ is affine in $\bs{u}$.

For a safe control scenario, we consider a payload insertion task inspired by the deployment of aerial payloads on natural surfaces. The robot must remain ${p_{x}^{\min} \!\geq\! 0}$ away from a wall while maintaining a minimum height and a bounded pitch angle. This is captured by the functions ${h_x(\bs{x}) \!=\! p_x \!- p_{x}^{\min}}$, ${h_z(\bs{x}) \!=\! p_z \! - p_{z}^{\min}}$, ${h_\vartheta(\bs{x}) \!=\! \vartheta_{\max}^2 \! -\vartheta^2}$, which are combined into a single safety function as
\begin{equation}
\label{eq:safety_function}
    h(\bs{x}) = 
    -\tfrac{1}{\kappa}
    \ln \! \big(
    {\rm e}^{-\kappa h_x(\bs{x})} + {\rm e}^{-\kappa h_z(\bs{x})} + {\rm e}^{-\kappa h_\vartheta(\bs{x})}
    \big),
\end{equation}
with a smoothing parameter ${\kappa > 0}$ and safe set $\Cs$ in \eqref{eq:safeset}.

The insertion task is executed by a primary controller $\up$ that is a cascaded position--attitude law steering the quadrotor to a goal ${(p_x^{\rm g},p_z^{\rm g})}$, saturated to $\mathcal{U}$. Since $\up$ is agnostic to both the safety constraint and the parametric uncertainty, we apply the proposed QP-based safety filter in \eqref{eq:adaptive_qp_main}: the DREM estimator of Example~\ref{ex:drem} provides the parameter estimate $\thetahat(t)$ and the certified error bounds $\rho_i(t)$, the backup flow is propagated with $\thetahat(t)$, and the safety conditions tightened by $\rho_i(t)$ are enforced along this flow. 

The robust backup controller takes the form 
\begin{equation}
\label{eq:quad_backup_controller}
\ub(\bs{x}) = \begin{bmatrix}
        F_{\max} &
        \operatorname{sat} \big( k_\vartheta (\vartheta - \vartheta_{\rm r}) + k_\omega\omega \big)
    \end{bmatrix}^\top,
\end{equation}
where ${k_\vartheta > 0}$ and ${k_\omega > 0}$ are gains, ${\vartheta_{\rm r} \in (0, \vartheta_{\max})}$ is a reference pitch angle, and $\operatorname{sat}$ is a smooth saturation map into $[-M_{\max}, M_{\max}]$. The corresponding backup set is
\begin{equation}
\label{eq:quad_backup_set}
\!\!\!\Cb \!\defeq\! 
\bigg\{\!
    \bs{x} \!\in\! \mathcal{X} \!:
    \begin{array}{@{}l@{}}
    h_{x,\rm b}(\bs{x}) \!\geq\! 0,
    \ h_{z,\rm b}(\bs{x}) \!\geq\! 0,
    \ h_{\vartheta,\rm b}(\bs{x}) \!\geq\! 0 \\
    h_{v_x,\rm b}(\bs{x}) \!\geq\! 0,
    \ h_{v_z,\rm b}(\bs{x}) \!\geq\! 0
    \end{array}
\!\bigg\}\!,
\end{equation}
where ${h_{x, \rm b}(\bs{x}) \!=\!  p_x \!-\! p_{x}^{\rm min} \!-\! r_x}$ and ${h_{z,\rm b}(\bs{x}) \!=\! p_z \!-\! p_{z}^{\rm min} \!-\! r_z}$ represent surface and altitude clearance with margins ${r_x \!>\! 0}$ and ${r_z \!>\! 0}$, while ${h_{v_x,\rm b}(\bs{x}) \!=\! v_x}$ and ${h_{v_z,\rm b}(\bs{x}) \!=\! v_z}$ prescribe nonnegative velocities. For the pitch dynamics, we let ${\bs{y}_\vartheta \!\defeq\!
\begin{bmatrix}
\vartheta - \vartheta_{\rm r} \!&\! \omega
\end{bmatrix}^{\top}}$
and ${h_{\vartheta,\rm b}(\bs{x})  \!=\! \varrho_\vartheta - \bs{y}_\vartheta^{\top} \bs{P}_\vartheta \bs{y}_\vartheta}$, with ${\varrho_\vartheta \!>\! 0}$ and ${\bs{P}_\vartheta \!=\! \bs{P}_\vartheta^{\top} \succ 0}$.
Further details on the construction of $\ub$ and $\Cb$ can be found in the supplementary material,
where it is shown that \eqref{eq:quad_backup_controller} renders $\Cb$ in \eqref{eq:quad_backup_set} robustly forward invariant for suitable choices of $\bs{P}_\vartheta$, $\varrho_\vartheta$, $r_x$, $r_z$, and $\vartheta_{\rm r}$.

Fig.~\ref{fig:results} shows a simulation of the planar quadrotor using the proposed robust adaptive bCBF and the DREM estimator of Example~\ref{ex:drem}. We compare our method against two baselines: the primary controller $\up$ alone, and the nominal bCBF method of Section~\ref{sec:pre} that propagates the backup flow with the nominal model (${\bs{\theta} = \bzero}$) and applies no tightening or robustification. The true parameters ${\bs{\theta} = \begin{bmatrix}0.08 \!\!& 0.08 \!\!& 0.22 \!\!& -0.37 \!\!& 0.008 \!\!& 0.003\end{bmatrix}^\top}$ correspond to a quadrotor that is approximately $59\%$ heavier than the nominal model; the remaining simulation parameters are provided in the code repository\footref{footnote:code}. As shown in Fig. \ref{fig:results}, the primary controller alone enters the unsafe set. The nominal bCBF intervenes late and violates the altitude clearance as the true quadrotor is heavier than the nominal model. On the other hand, the proposed robust adaptive bCBF ensures safety by keeping the quadrotor within $\mathcal{C}_{\rm S}$ and satisfying the input constraints despite parametric uncertainty. This is achieved using the parameter estimates, whose certified bounds shrink over time. The estimated backup trajectories $\varphibhat{\tau}{\bs{x}}$ initially deviate from the true backup trajectories $\varphibtrue{\tau}{\bs{x}}$ and later coincide with them as the estimates converge, making the tightening terms vanish and the inner approximation $\Cdhat$ approach the implicit safe set $\mathcal{C}_{\rm I}$.

\section{Conclusion}
\label{sec:conc}
We presented a robust adaptive backup CBF framework for the safety-critical control of nonlinear systems with parametric uncertainty and bounded inputs. This method computes the backup flow using an online parameter estimate and tightens the backup safety constraints with certified component-wise error bounds. A duality-based robustification yielded a safety filter in the form of a convex QP. Future work will focus on feasibility certificates for the QP in (35) under input constraints and tighter bounds for flow prediction errors.   

\appendix
\label{appendix}
We derive the constraints in \eqref{eq:uncertain_bCBF_constraint} by expressing the time derivatives $\dot{\bar h}$ and $\dot{\bar h}_{\rm b}$ in \eqref{eq:sufficient_cond}.
First, we define the derivatives
\begin{equation*}
\begin{aligned}
    \Phihat
    \defeq
    \derp{\varphibhat{\tau}{\bs{x}}}{\bs{x}},
    \ \ 
    \Gammahat
    \defeq
    \derp{\varphibhat{\tau}{\bs{x}}}{\thetahat}, \\
    \hat{\bs{J}}_{\rm b}(\tau,\bs{x},t)
    \defeq
    \left.
    \derp{}{\bs{z}}
    \bs{f}_{\rm b}\big(\bs{z},\thetahat(t)\big)
\right|_{\bs{z}=\varphibhat{\tau}{\bs{x}}},
\end{aligned}
\end{equation*}
where $\bs{\Phi}$ and $\bs{\Gamma}$ represent the sensitivity of the estimated backup flow $\widehat{\bs{\varphi}}_{\rm b}$ to the current state and parameter estimate, while $\hat{\bs{J}}_{\rm b}$ is the Jacobian of the estimated backup system~\eqref{eq:estimated_backup_flow_main}.
The sensitivity matrices can be computed by solving
\begin{equation*}
\begin{aligned}
    \tfrac{\partial}{\partial \tau}\bs{\Phi}(\tau,\bs{x},t)
    &=
    \hat{\bs{J}}_{\rm b}(\tau,\bs{x},t)
    \bs{\Phi}(\tau,\bs{x},t), \quad
    \bs{\Phi}(0,\bs{x},t)
    =
    \eye,
    \nonumber\\
    \tfrac{\partial}{\partial \tau}\bs{\Gamma}(\tau,\bs{x},t)
    &=
    \hat{\bs{J}}_{\rm b}(\tau,\bs{x},t)
    \bs{\Gamma}(\tau,\bs{x},t)
    +
    \bs{\psi}_{\rm b}(\varphibhat{\tau}{\bs{x}}),
    \\
    \bs{\Gamma}(0,\bs{x},t)
    &=
    \bzero.
    \nonumber
\end{aligned}
\end{equation*}

Then, by differentiating the expressions of $\bar{h}$ and $\bar{h}_{\rm b}$ in \eqref{eq:barriers_subset}, the total time derivatives $\dot{\bar h}$ and $\dot{\bar h}_{\rm b}$ become
\begin{equation*}
\begin{aligned}
    \dot{\bar h}(\tau,\bs{x},t,\bs{u})
    & = \derp{\bar{h}}{\bs{x}}(\tau,\bs{x},t) \dot{\bs{x}} + \derp{\bar{h}}{t}(\tau,\bs{x},t), \\
    \dot{\bar h}_{\rm b}(\bs{x},t,\bs{u})
    & = \derp{\bar{h}_{\rm b}}{\bs{x}}(\bs{x},t) \dot{\bs{x}} + \derp{\bar{h}_{\rm b}}{t}(\bs{x},t),
\end{aligned}
\end{equation*}
where the partial derivatives are expressed via the chain rule:
\begin{equation*}
\begin{aligned}
    \derp{\bar{h}}{\bs{x}}(\tau,\bs{x},t)
    &=
    \nabla h(\varphibhat{\tau}{\bs{x}})\Phihat
    -
    \derp{\epsilon}{\bs{x}}(\tau,\bs{x},t),
    \\
    \derp{\bar{h}}{t}(\tau,\bs{x},t)
    &=
    \nabla h(\varphibhat{\tau}{\bs{x}})
    \Gammahat\dot{\thetahat}
    -
    \derp{\epsilon}{t}(\tau,\bs{x},t),
    \end{aligned}
    \end{equation*}
    \begin{equation*}
    \begin{aligned}
    \derp{\bar{h}_{\rm b}}{\bs{x}}(\bs{x},t)
    &=
    \nabla h_{\rm b}(\varphibhat{T}{\bs{x}})\PhihatT
    -
    \derp{\epsilon_{\rm b}}{\bs{x}}(\bs{x},t),
    \\
    \derp{\bar{h}_{\rm b}}{t}(\bs{x},t)
    &=
    \nabla h_{\rm b}(\varphibhat{T}{\bs{x}})
    \GammahatT\dot{\thetahat}
    -
    \derp{\epsilon_{\rm b}}{t}(\bs{x},t).
\end{aligned}
\end{equation*}
Here $\dot{\thetahat}$ is given by the estimator in \eqref{eq:theta_adaptation_law},
while $\dot{\bs{x}}$ is given by the dynamics in \eqref{eq:unc_sys_affine}  as
\begin{equation*}
    \dot{\bs{x}}
    =
    \bs{f}(\bs{x})
    +
    \bs{g}(\bs{x})\bs{u}
    +
    \bs{\phi}(\bs{x},\bs{u})\thetahat(t)
    +
    \bs{\phi}(\bs{x},\bs{u})\thetatilde(t).
\end{equation*}

Next, substituting these derivatives into \eqref{eq:sufficient_cond} and moving all terms to the left-hand side, we obtain \eqref{eq:uncertain_bCBF_constraint} with
\begin{equation*}
\begin{aligned}
    a(\tau,\bs{x},\bs{u},t) & =
    \derp{\bar{h}}{\bs{x}}(\tau,\bs{x},t)
    \big( \bs{f}(\bs{x}) + \bs{g}(\bs{x})\bs{u} + \bs{\phi}(\bs{x},\bs{u})\thetahat(t) \big)
    \\
    & \quad
    + \derp{\bar{h}}{t}(\tau,\bs{x},t)
    + \alpha\big(\bar h(\tau,\bs{x},t)\big),
    \\
    \bs{c}(\tau,\bs{x},\bs{u},t)^{\top} & =
    \derp{\bar{h}}{\bs{x}}(\tau,\bs{x},t)
    \bs{\phi}(\bs{x},\bs{u}),
    \\
    a_{\rm b}(\bs{x},\bs{u},t) & =
    \derp{\bar{h}_{\rm b}}{\bs{x}}(\bs{x},t)
    \big( \bs{f}(\bs{x}) + \bs{g}(\bs{x})\bs{u} + \bs{\phi}(\bs{x},\bs{u})\thetahat(t) \big)
    \\
    & \quad
    + \derp{\bar{h}_{\rm b}}{t}(\bs{x},t)
    + \alpha_{\rm b}\big(\bar h_{\rm b}(\bs{x},t)\big),
    \\
    \bs{c}_{\rm b}(\bs{x},\bs{u},t)^{\top} & =
    \derp{\bar{h}_{\rm b}}{\bs{x}}(\bs{x},t)
    \bs{\phi}(\bs{x},\bs{u}).
\end{aligned}
\end{equation*}

\bibliographystyle{IEEEtran}
\bibliography{references.bib} 

\section*{Supplementary Material: Robust Backup Controller and Backup Set for the Uncertain Planar Quadrotor Model}

To simplify the notation used in this section, we use ${\lambda_{\min}(\bs{S})}$ and ${\lambda_{\max}(\bs{S})}$ to denote the smallest and largest eigenvalues of a symmetric matrix $\bs{S}$, and we define
\begin{equation*}
\begin{aligned}   
    & a_m(\bs{\theta}) \defeq \frac{1}{m_0} + \delta_m,
    \ \ 
    a_J(\bs{\theta}) \defeq \frac{1}{J_0} + \delta_J, 
    \ \ \bs{e}_1 \defeq \begin{bmatrix} 1 & 0 \end{bmatrix}^{\top},
    \\ 
    &  \bs{K}_\vartheta
    \!\defeq\! 
    \begin{bmatrix}
        k_\vartheta & k_\omega
    \end{bmatrix},  
    \ \bs{b}_\vartheta \!\defeq\! \begin{bmatrix} 0 & 1 \end{bmatrix}^{\top}, 
     \
    \bs{A}_\vartheta(a) \!\defeq\!\!
    \begin{bmatrix}
        0 & \!1 \\
        -a k_\vartheta & \!-a k_\omega
    \end{bmatrix},
    \\
    & \Delta_\vartheta \!\defeq\!\! \sqrt{\!\varrho_\vartheta \bs{e}_1^{\top} \bs{P}_\vartheta^{-1} \bs{e}_1},
    \ 
    \Delta_M \!\defeq\!\! \sqrt{\!\varrho_\vartheta\bs{K}_\vartheta \bs{P}_\vartheta^{-1} \bs{K}_\vartheta^{\top}},
    \  
    \bar \vartheta \!\defeq\! \vartheta_{\rm r} \!+\! \Delta_\vartheta.
\end{aligned}
\end{equation*}
Then, parameter bounds in Assumption~\ref{assump:theta_initial_box} induce known constants satisfying
\begin{equation*}
\begin{aligned}  
   & \delta_g^{-} \leq  \delta_g \leq \delta_g^{+},
    \ \
    \delta_\ell^{-} \leq \delta_\ell \leq \delta_\ell^{+}, 
    \\
    & 0 < a_J^{-} \leq a_J(\bs{\theta}) \leq a_J^{+},
    \ \ 0 < a_m^{-} \leq a_m(\bs{\theta}) \leq a_m^{+},
    \end{aligned}
\end{equation*}
for all ${\bs{\theta} \in \Theta}$. The proposition below provides sufficient conditions on the design parameters of the backup controller \eqref{eq:quad_backup_controller} and the backup set \eqref{eq:quad_backup_set} for Assumption~\ref{assump:robust_backup_main} to hold for the uncertain quadrotor model \eqref{eq:uncertain_drone_dyn}.
\begin{proposition} 
\label{prop:quad_backup_invariance} 
Consider the uncertain quadrotor model  \eqref{eq:uncertain_drone_dyn}, the safety function \eqref{eq:safety_function} with the safe set $\Cs$, the backup controller \eqref{eq:quad_backup_controller}, and the backup set $\Cb$ in \eqref{eq:quad_backup_set}. If the controller gains ${k_\vartheta > 0}$, ${k_\omega > 0}$ and the design parameters ${\bs{P}_\vartheta = \bs{P}_\vartheta^{\top} \succ 0}$, ${\bs{Q}_\vartheta = \bs{Q}_\vartheta^{\top} \succ 0}$, ${\varrho_\vartheta > 0}$, ${\vartheta_{\rm r} \in (0,\vartheta_{\max})}$, ${r_x, r_z > 0}$, and ${r_\vartheta \in (0,\vartheta_{\max}^2)}$, with ${\bar\vartheta < \pi/2}$, are chosen such that
\begin{equation}
\label{eq:quad_backup_conditions}
\begin{aligned} 
& \bs{A}_\vartheta(a_J^{-})^{\top} \bs{P}_\vartheta + \bs{P}_\vartheta \bs{A}_\vartheta(a_J^{-}) \preceq -\bs{Q}_\vartheta, \\
& \bs{A}_\vartheta(a_J^{+})^{\top} \bs{P}_\vartheta + \bs{P}_\vartheta \bs{A}_\vartheta(a_J^{+}) \preceq -\bs{Q}_\vartheta, \\
& \Delta_M  \leq M_{\max}, \\
& \vartheta_{\rm r} - \Delta_\vartheta  \geq 0, \\
& \vartheta_{\rm r} + \Delta_\vartheta  \leq \sqrt{\vartheta_{\max}^2 - r_\vartheta}, \\
& \varrho_\vartheta
     \!\geq\!
    \lambda_{\max}(\bs{P}_\vartheta) \!
    \left(\!
        \frac{
        2\norm{\bs{P}_\vartheta\bs{b}_\vartheta}
        F_{\max}
        \max\{|\delta_\ell^{-}|,|\delta_\ell^{+}|\}
        }{
        \lambda_{\min}(\bs{Q}_\vartheta)
        }
    \!\right)^2, \\
   &  a_m^{-} F_{\max} \cos(\bar\vartheta) - g_0 -\delta_g^{+}  \geq 0, \\
   & {\rm e}^{-\kappa r_x} + {\rm e}^{-\kappa r_z} + {\rm e}^{-\kappa r_\vartheta}  \leq 1,
\end{aligned}
\end{equation}
then the moment saturation in \eqref{eq:quad_backup_controller} is inactive on $\Cb$, ${\Cb \subseteq \Cs}$ holds, and the backup controller \eqref{eq:quad_backup_controller} renders $\Cb$ robustly forward invariant for \eqref{eq:uncertain_drone_dyn} for all ${\bs{\theta} \in \Theta}$, i.e., Assumption~\ref{assump:robust_backup_main} is satisfied.
\end{proposition}

\begin{proof}
We first consider the pitch dynamics. On the set $\Cb$, ${h_{\vartheta,\rm b}(\bs{x}) \geq 0}$ implies ${\bs{y}_\vartheta^{\top} \bs{P}_\vartheta \bs{y}_\vartheta \leq \varrho_\vartheta}$, and hence
\begin{equation*}
    |\vartheta - \vartheta_{\rm r}| \leq \Delta_\vartheta, \ \ \left| \bs{K}_\vartheta \bs{y}_\vartheta \right| \leq \Delta_M.
\end{equation*}
Thus, by the fourth and fifth conditions in \eqref{eq:quad_backup_conditions}, the pitch angle satisfies ${0 \leq \vartheta \leq \bar\vartheta \leq \sqrt{\vartheta_{\max}^2 - r_\vartheta}}$ on $\Cb$. Moreover, the third condition in \eqref{eq:quad_backup_conditions} yields ${|\bs{K}_\vartheta \bs{y}_\vartheta| \leq M_{\max}}$, and hence the moment saturation in \eqref{eq:quad_backup_controller} is inactive on $\Cb$.

We next verify the robust invariance of the pitch dynamics under the backup controller, which satisfy
\begin{equation*}
    \dot{\bs{y}}_\vartheta =
    \bs{A}_\vartheta (a_J(\bs{\theta})) \bs{y}_\vartheta
    + \bs{b}_\vartheta \delta_\ell F_{\max}.
\end{equation*}
Since $\bs{A}_\vartheta(a)$ depends affinely on $a$, the first two conditions in \eqref{eq:quad_backup_conditions} imply
\begin{equation*}
    \bs{A}_\vartheta(a)^{\top} \bs{P}_\vartheta +
    \bs{P}_\vartheta\bs{A}_\vartheta(a) \preceq -\bs{Q}_\vartheta,
    \ \  \forall a \in[a_J^{-}, a_J^{+}].
\end{equation*}
Therefore, with ${V_\vartheta \defeq \bs{y}_\vartheta^{\top} \bs{P}_\vartheta \bs{y}_\vartheta}$, we have
\begin{equation*}
\begin{aligned}
 & \dot V_\vartheta
      =
    \bs{y}_\vartheta^{\top}
    \!\left(
        \bs{A}_\vartheta(a_J(\bs{\theta}))^{\top}\bs{P}_\vartheta
        \!+\!
        \bs{P}_\vartheta\bs{A}_\vartheta(a_J(\bs{\theta}))
    \!\right)\!
    \bs{y}_\vartheta
    \\
    & \ \ \ \ +
    2 \bs{y}_\vartheta^{\top} \bs{P}_\vartheta \bs{b}_\vartheta\delta_\ell F_{\max}
    \\
    & \!\leq \!\!
    -\lambda_{\min}(\bs{Q}_\vartheta) \! \norm{\bs{y}_\vartheta}^{2}
    \!\!+\!
    2 F_{\max} \! \norm{\bs{P}_\vartheta\bs{b}_\vartheta}
    \max\{|\delta_\ell^{-}|,|\delta_\ell^{+}|\}
    \! \norm{\bs{y}_\vartheta}.
\end{aligned}
\end{equation*}
On the boundary ${h_{\vartheta,\rm b}(\bs{x}) = 0}$, we have ${V_\vartheta = \varrho_\vartheta}$, and hence ${\norm{\bs{y}_\vartheta} \geq \sqrt{\varrho_\vartheta / \lambda_{\max}(\bs{P}_\vartheta)} > 0}$. Therefore, by the sixth condition in \eqref{eq:quad_backup_conditions}, ${\dot V_\vartheta \leq 0}$ on the boundary of the zero superlevel set of $h_{\vartheta,\rm b}(\bs{x})$. Equivalently, ${\dot h_{\vartheta,\rm b}(\bs{x}) \geq 0}$ on this boundary, and the pitch component of $\Cb$ is robustly forward invariant.

We now consider the horizontal component of the motion. Since ${0 \leq \vartheta \leq \bar\vartheta < \pi/2}$ on $\Cb$, the horizontal dynamics of \eqref{eq:uncertain_drone_dyn} under \eqref{eq:quad_backup_controller} satisfy
\begin{equation*}
    \dot v_x =
    a_m(\bs{\theta}) F_{\max} \sin(\vartheta) - c_x v_x.
\end{equation*}
On the boundary ${h_{v_x, \rm b}(\bs{x}) \!=\! 0}$, we have ${v_x = 0}$, and therefore
\begin{equation*}
    \dot h_{v_x,\rm b}(\bs{x}) =
    a_m(\bs{\theta}) F_{\max} \sin(\vartheta) \geq 0,
\end{equation*}
where we used ${a_m(\bs{\theta}) > 0}$ and ${\vartheta \geq 0}$. Hence, the $v_x$ component of $\Cb$ is robustly forward invariant. Moreover, on $\Cb$, ${v_x \geq 0}$, and hence ${\dot h_{x,\rm b}(\bs{x}) = v_x \geq 0}$. Therefore, the surface clearance component of $\Cb$ is robustly forward invariant.

We next consider the altitude component of the motion. On the boundary ${h_{v_z,\rm b}(\bs{x}) = 0}$, we have ${v_z = 0}$, and the vertical dynamics of \eqref{eq:uncertain_drone_dyn} satisfy
\begin{equation*}
\begin{aligned}
    \dot h_{v_z,\rm b}(\bs{x})  = \dot v_z
    & = -g_0 - \delta_g + a_m(\bs{\theta})F_{\max}\cos(\vartheta)
    \\
    & \geq -g_0 - \delta_g^{+} + a_m^{-}F_{\max}\cos(\bar\vartheta)
    \geq 0,
\end{aligned}
\end{equation*}
where we used ${0 \leq \vartheta \leq \bar \vartheta < \pi/2}$ and the seventh condition in \eqref{eq:quad_backup_conditions}. Therefore, the $v_z$ component of $\Cb$ is robustly forward invariant. Furthermore, on $\Cb$, ${v_z \geq 0}$, thus ${\dot h_{z,\rm b}(\bs{x}) = v_z \geq 0}$, and the altitude clearance component of $\Cb$ is also robustly forward invariant.

Finally, we show that ${\Cb \subseteq \Cs}$. Since ${h_{x,\rm b}(\bs{x}) \geq 0}$ and ${h_{z,\rm b}(\bs{x}) \geq 0}$, we have ${h_x(\bs{x}) \geq r_x}$ and ${h_z(\bs{x}) \geq r_z}$. Moreover, by the fifth condition in \eqref{eq:quad_backup_conditions}, ${h_{\vartheta,\rm b}(\bs{x}) \geq 0}$ implies ${h_\vartheta(\bs{x}) = \vartheta_{\max}^2-\vartheta^2 \geq r_\vartheta}$. Therefore, by the last condition in \eqref{eq:quad_backup_conditions}, the smooth minimum safety function \eqref{eq:safety_function} satisfies ${h(\bs{x}) \geq 0}$ for all ${\bs{x} \in \Cb}$, i.e., ${\Cb \subseteq \Cs}$.

Combining the pitch, horizontal, and altitude invariance results with the input admissibility and ${\Cb \subseteq \Cs}$, the backup controller \eqref{eq:quad_backup_controller} renders the backup set \eqref{eq:quad_backup_set} robustly forward invariant for \eqref{eq:uncertain_drone_dyn} for all ${\bs{\theta} \in \Theta}$.
\end{proof}

\end{document}